\documentclass[11pt,letter]{article}
\usepackage{microtype}
\usepackage{amsmath,amsfonts,amssymb,amsthm}
\usepackage{algorithm,algorithmic}
\usepackage{placeins}
\usepackage[margin=1in]{geometry}
\newtheorem{lemma}{Lemma}
\newtheorem{theorem}{Theorem}
\newtheorem{proposition}{Proposition}
\begin{document}
\title{A Refined Analysis of LSH for Well-dispersed Data Points}
\author{Wenlong Mou\thanks{Key Laboratory of Machine Perception, School of EECS, Peking University. Email: \texttt{mouwenlong@pku.edu.cn}} \\
\and
Liwei Wang\thanks{Key Laboratory of Machine Perception, School of EECS, Peking University. Email: \texttt{wanglw@cis.pku.edu.cn}}}
\date{\today}

\maketitle
	
	\begin{abstract}
		Near neighbor problems are fundamental in algorithms for high-dimensional Euclidean spaces. While classical approaches suffer from the curse of dimensionality, locality sensitive hashing (LSH) can effectively solve $\alpha$-approximate $r$-near neighbor problem, and has been proven to be optimal in the worst case. However, for real-world data sets, LSH can naturally benefit from well-dispersed data and low doubling dimension, leading to significantly improved performance.\par
		In this paper, we address this issue and propose a refined analyses for running time of approximating near neighbors queries via LSH. We characterize dispersion of data using $N_{\beta}$, the number of $\beta r$-near pairs among the data points. Combined with optimal data-oblivious LSH scheme, we get a $O\left(\left(1+\frac{4\sqrt{2}\alpha}{\beta}\right)^{\frac{d}{2\alpha^2}}(n+N_{\beta})^{\frac{1}{2\alpha^2}}\right)$ bound for expected query time. For many natural scenarios where points are well-dispersed or lying in a low-doubling-dimension space, our result leads to sharper performance than existing worst-case analysis. This paper not only presents the \emph{first} rigorous proof on how LSHs make use of the structure of data points, but also provides important insights into parameter setting in the practice of LSH beyond worst case. Besides, the techniques in our analysis involve a generalized version of sphere packing problem, which might be of some independent interest.
	\end{abstract}
	
	\section{Introduction}
	Near neighbor search is a fundamental problems in metric spaces, and is playing an increasingly important role in databases~\cite{DBLP:journals/pvldb/SunWQZL14}, machine learning~\cite{DBLP:conf/nips/Shrivastava014} and computer vision~\cite{DBLP:conf/cvpr/Ben-ZrihemZ15}. With the large-scale data set, we usually need a data structure with sub-linear query time, which only visits a small portion of candidates. There are many classical results on near neighbor search in Euclidean spaces with fixed dimensions. Voronoi diagrams partition the space based on nearest neighbors, but the computation of which is formidable for high-dimensional spaces. Tree structures such as k-d trees~\cite{DBLP:journals/cacm/Bentley75} and VP trees~\cite{DBLP:conf/soda/Yianilos93} are effective for low dimensional spaces. However, those data structures usually suffer heavily from the curse of dimensionality, and could not be put into practical use for the emerging large-scale data sets. There are also many works on generalizing those tree structures to high dimensional cases via low-distortion dimensionality reduction~\cite{DBLP:conf/compgeom/Anagnostopoulos15,DBLP:journals/talg/IndykN07} and doubling dimensions~\cite{randapproxdoubling}.\par
	Though it seems hard to derive deterministic algorithms to the exact problems, randomization and approximation allow effective solutions. The approximate randomized formulation was proposed in~\cite{DBLP:conf/stoc/IndykM98,Andoni09thesis} as follows:\par
	\begin{description}
		\item[Randomized $\alpha$-approximate $r$-near neighbor] Given $S=\{x_1,x_2,\ldots, x_n\}\subset \mathbb{R}^d$, $r\in \mathbb{R}^+$, $\alpha>1, \delta\in(0,1)$, construct a data structure such that given any query point $x_0$, if there is some $x_{i}\in S$ s.t. $\Vert x_0-x_i\Vert_2\leq r$, then we can report a point $x_j\in S$ s.t. $\Vert x_0-x_j\Vert_2\leq \alpha r$ with probability at least $1-\delta$.
	\end{description}
	\par
	One of the most successful solutions to this problem was locality sensitive hashing(LSH), which constructs a distribution on hash functions, and makes near neighbors more likely to be hashed into the same bucket. In each query, we only need to visit the buckets where query point have ever been put. So we only need to visit a small proportion of data points in total, and that will guarantee essentially sub-linear query time. Basically, the more sharply collision probability decreases with distance, the better performance will be obtained. Let $p_1$ be the collision probability for points at distance $r$, and $p_2$ be that for points at distance $\alpha r$. The key performance measure for LSH is $\rho=\frac{\log p_1}{\log p_2}$, since $\alpha$-approximate $r$-near neighbor will be solved within query $O\left(n^{\rho}\right)$, with proper parametric setting. The LSHs for Euclidean spaces have been intensively studied, and a series of hashing schemes were proposed with improving $\rho$ parameter~\cite{DBLP:conf/compgeom/DatarIIM04,DBLP:conf/focs/AndoniI06,DBLP:conf/soda/AndoniINR14,DBLP:conf/stoc/AndoniR15}.\par
	Despite their success, however, the classical analyses for locality sensitive hashing methods are not always optimal: they simply treat all distant points as the same. One can see from a simple geometric intuition, that when the number of near pairs is controlled, then the data points have to be dispersed. In this paper, we quantify this idea and present a  refined analysis for LSH. Concretely, we introduce a new parameter as $N_{\beta}=\big|\left\{(x_i,x_j):\Vert x_i-x_j\Vert_2\leq \beta r\right\}\big|$ and show how the performance of LSH can be improved if $N_{\beta}$ increases slowly with $\beta$. In our analyses, we require a slightly stronger uniform LSH condition, namely, the collision probability bound to hold uniformly for any $\alpha>1$. It is easy to see that this condition is satisfied by all existing data-independent LSHs. To characterize the points that are more than $\alpha r$ far from $x_0$, we propose a generalized version of the sphere packing problem and give an upper bound. We proved that for any $\beta>0$, every uniform LSH with $\rho=\rho(\alpha)$ for $\mathbb{R}^d$ guarantees a $O\left(\left(1+O(1)\cdot\frac{\alpha}{\beta}\right)^{\frac{d\rho(\alpha)}{2}}(n+N_{\beta})^{\frac{\rho(\alpha)}{2}}\right)$ query time, by setting the parameters properly. We also show that our dimensionality-dependent analysis for LSH can be easily generalized to the case of doubling metrics, at a loss of constant factor. Compared with classical analyses of LSH, our bounds achieved essentially better performance in two cases: (i) when the dimensionality or the doubling dimension is low (ii) when $N_{\beta}$ remains at $O(n)$ for some large constant $\beta$. In addition to a good explanation on how LSH works better than worst-case analyses on real-world data, this bound can provide more insights for the choice of parameters for LSH. To the best of our knowledge, this is the first theoretical explanation on how LSHs exploit the intrinsic characteristics of data points.\par
	\section{LSH for near neighbor problems}
	In this section, we will give a brief introduction to the formulation of LSH and query algorithms. A precise analysis on the expected query time based on uniform LSH will be presented. We will also summarize the existing works on LSHs for Euclidean spaces.\par
	LSH is formulated as a distribution over hashing functions, for which the probability of collision increases as two points get closer. In the original formulation of LSH, parameter $\alpha$ and $r$ are fixed, and only the collision probability for $\Vert x-y\Vert_2\leq r$ and $\Vert x-y\Vert_2\geq \alpha r$ is considered. A key property in the evaluation of LSH is $\rho=\frac{\log p_1}{\log p_2}$, which depicts how sharply the collision probability decreases with distance. As discussed in \cite{DBLP:conf/focs/AndoniI06,DBLP:conf/soda/AndoniINR14}, we usually need to make some tradeoff between computational cost of $h(\cdot)$ itself and accuracy of LSH. So the actual $\rho$-parameter for LSH is often written as $\tilde{\rho}(\alpha)+o(1)$, where the residual term $o(1)$ term diminishes as $n\rightarrow\infty$, though the ratio between log probability $\tilde{\rho}$ for ideal hashing class is independent of $n$. We will also use this notation.
	\begin{description}
		\item[Locality sensitive hashing]:
		A distribution $\mathcal{H}$ is called $\langle r,\alpha r,p_1,p_2\rangle$-sensitive if for $\forall x,y\in \mathbb{R}^d$:
		\begin{itemize}
			\item $\forall p,q\in \mathbb{R}^d, \Vert p-q\Vert_2\leq r$ we have $Pr_{h\sim \mathcal{H}}\left\{h(p)=h(q)\right\}\geq p_1$
			\item $\forall p,q\in \mathbb{R}^d, \Vert p-q\Vert_2\geq \alpha r$ we have $Pr_{h\sim \mathcal{H}}\left\{h(p)=h(q)\right\}\leq p_2$
		\end{itemize}
	\end{description}
	We slightly strengthen the requirements for LSH in our analysis, where the collision probability gap should be guaranteed uniformly for $\forall \alpha>1$. Fortunately, all the known data-independent LSH do not depend upon parameter $\alpha$ and are suitable for this formulation. In the rest of this paper, we will denote uniform locality sensitive hashings using the term LSH, except when it is specified as data-dependent.
	\begin{description}
		\item[Uniform locality sensitive hashing] A distribution $\mathcal{H}$ is called uniformly $\langle r,\rho=\rho(s)\rangle$-sensitive if for $\forall x,y\in \mathbb{R}^d$ satisfies:\\
		$p(s)\triangleq\mathop{Pr}_{h\sim\mathcal{H}}\left\{h(x)=h(y)\big|\Vert x-y\Vert_2=s\cdot r\right\}$ is a monotonic decreasing function of $s\in(1,+\infty)$ and $\rho(s)=\frac{\log p(1)}{\log p(s)}$.
	\end{description}
	The general algorithm framework for LSH was proposed in \cite{DBLP:conf/stoc/IndykM98,Andoni09thesis}, as described in Algorithm~\ref{framework}.
	\begin{algorithm}[htb]
		\caption{Framework for near neighbor approximation}\label{framework}
		\textbf{Input}: $x_1,x_2,\ldots,x_n\in\mathbb{R}^d,\quad\alpha>1,\quad r\in R^+$\\
		\textbf{Parameters}: $K, L$
		\hrule\vspace{0.12cm}
		\textbf{Preprocessing}:\\
		Sample $L\cdot K$ functions $h_{11},h_{12},\ldots,h_{1k},h_{21},\ldots,h_{2k},\ldots,h_{LK}\sim i.i.d.\mathcal{H}$ and let $g_i=(h_{i1},h_{i2},\ldots,h_{iK}),\quad \forall i=1,2,.\ldots L$.\\
		Construct $L$ hash tables with respect to $\left\{g_i\right\}_{i=1}^{L}$.
		\hrule\vspace{0.12cm}
		\textbf{Query}:\\
		Given query point $x_0$:\\
		For $i=1,2,\ldots,L$:
		\begin{itemize}
			\item Compute $g_i(x_0)$ and locate the hashing bucket.
			\item Traverse the elements in bucket and compute the actual distances, if any $\alpha r$ near neighbor is found, report it and stop.
		\end{itemize}
	\end{algorithm}
	
	The parameter $\langle K,L\rangle$ will be chosen carefully to optimize the performance while guaranteeing success probability. If there is a neighbor $x^*$ within distance $r$, to guarantee $\alpha r$-near approximation, we just need to visit $x^*$ with $(1-\delta)$ probability:
	\begin{equation}\mathop{Pr}_{g\sim\mathcal{H}^{KL}}\left\{\forall g_i, g_i(x_0)\neq g_i(x^*)\right\}\leq \left(1-p(1)^K\right)^L\leq \delta\end{equation}
	So it is sufficient to set the number of rounds run for each query as $L=p(1)^{-K}\log \frac{1}{\delta}$.\par
	Consider the expected number of points visited in each round of hashing computation. If a point at distance at most $\alpha r$ was ever visited, the algorithm would stop and report this point. So at most one point at distance within $\alpha r$ was visited in each round, and other visited points are located farther than $\alpha r$. Let the set of points visited in the $i$-th round be $S_i$.
	\begin{equation}
		\begin{split}
			&\mathop{E}_{h\sim\mathcal{H}}\left[\vert S_i\vert\right]\\
			\leq &1+\sum_{\Vert x_i-x_0\Vert_2\geq \alpha r}\mathop{Pr}_{h\sim\mathcal{H}}\left\{h(x_i)=h(x_0)\right\}^K\\
			=&1+\sum_{\Vert x_i-x_0\Vert_2\geq \alpha r}\exp\left\{\frac{K\cdot\log p(1)}{\rho\left(\frac{\Vert x_i-x_0\Vert_2}{r}+o(1)\right)}\right\}
		\end{split}
	\end{equation}
	In the classical analysis for locality sensitive hashing, we usually relax the inequality by plugging in $\Vert x_i-x_0\Vert_2\geq \alpha r$, and get:
	\begin{equation}
	\mathop{E}_{h\sim\mathcal{H}}\left[\vert S_i\vert\right]\leq 1+ n\cdot p(1)^{\frac{K}{\rho(\alpha)+o(1)}}
	\end{equation}
	However, in many cases, this relaxation cannot be tight simultaneously for points in $S$, and some geometric constraints will force some of the points farther away from $x_0$. This phenomenon is the main focus of this paper, and will be analyzed in detail in the next sections. Here we just use this relaxation and go through the main results by classical analysis of LSH.\par
	By putting all the above together, we can bound the expectation of running time for near neighbor approximation with: (Let $\tau$ be the time required for computing LSH)
	\begin{equation}T(n)=O\left(p(1)^{-K} \left(1+n\cdot p(1)^{\frac{K}{\rho(\alpha)+o(1)}}\right)\tau\log\frac{1}{\delta}\right)\end{equation}
	By choosing the optimal parameter $K$, we get an upper bound for the query running time: $T_{query}(n)=O\left(n^{\rho(\alpha)+o(1)}\tau\log\frac{1}{\delta}\right)$ and the corresponding preprocessing time can be upper bounded with $T_{preprocess}(n)=O\left(n^{1+\rho(\alpha)+o(1)}\tau\log\frac{1}{\delta}\right)$.\par
	Locality sensitive hashing and near neighbor problem have been intensively studied in existing literature. In \cite{DBLP:conf/compgeom/DatarIIM04} a class of locality sensitive hashing was firstly proposed. They obtained a  $\rho(\alpha)=\frac{1}{\alpha}$ performance by projecting the data points to a calibrated real line. Later, a significant improvement was done in \cite{DBLP:conf/focs/AndoniI06}. They first perform random projection and reduce to a low-dimensional space, then the hashing buckets were constructed using random grids of balls. By setting appropriate parameters, they achieved the performance $\rho(\alpha)=\frac{1}{\alpha^2}+o(1)$, asymptotically. As shown in \cite{DBLP:journals/toct/ODonnellWZ14}, this bound is essentially optimal in the worst case.\par 
	Recently, there is also a series of works on data-dependent locality sensitive hashing. Andoni et al., \cite{DBLP:conf/soda/AndoniINR14} first introduced a class of data dependent hashing class, and improved the $\rho$ parameter to $\frac{7}{8\alpha^2}+\frac{O(1)}{\alpha^3})$. In \cite{DBLP:conf/stoc/AndoniR15} it was further improved to $\frac{1}{2\alpha^2-1}+o(1)$, which is proven to be optimal in \cite{DBLP:journals/corr/AndoniR15a}. Unfortunately, since their construction of hashing schemes depend on parameter $\alpha$, they could not be generalized to the uniform LSH case, and those bounds are not suitable for our refined analyses.\par
	\section{Generalized Sphere Packing Problem}
	In the classical analysis of locality sensitive hashing, we relax the estimation for query time by assuming all the data points visited but not accepted in a query are just $\alpha r$ far away from the query point. But this is not usually true in reality: if most of the data points approximately on a sphere with radius $\alpha r$ centered at $x_0$, then the sphere will be "crowded" and there will be many pairs of near neighbors within the data set. In reality, however, most data points in $S$ are far from each other. For parameter $\beta>1$, a fixed constant that is not very large, we will have $N_{\beta}=\big|\left\{(x_i,x_j):\Vert x_i-x_j\Vert_2\leq \beta r\right\}\big|\ll n^2$. So in our refined analysis, we seek to bound the running time of algorithm in terms of not only $n=\big|S\big|$ but also $N_{\beta}$, where $\beta$ is a parameter used for minimize the bound depending on the structure of data. In the following analysis, we make use of a generalized version of famous sphere packing problem: we want to characterize the phenomenon that a set of points must be well-dispersed if there are only a few near pairs.\par
	An intuitive view to this problem is to consider the "worst case", where the two data points either coincide, or be closely packed in $\mathbb{R}^d$ with distance at least $\beta r$. We construct a graph where two vertices are linked if their corresponding points are $\beta r$ near neighbors. Roughly speaking, we want to show that by adjusting the configuration of points into the "worst-case", we will shrink the space those points take without increasing the number of near pairs. The following lemma gives us a quantitative description of this "worst-case" intuition.
	\begin{lemma}\label{graphshrink}
		For an undirected graph $G=\langle V,E\rangle$, with $\big|V\big|=n$. There is a subset of vertices $T\subset V$ and a mapping $\phi: V\rightarrow T$, such that
		\begin{itemize}
			\item $\forall u\in T, \phi(u)=u$ and $\forall v\in V-T, (v,\phi(v))\in E$
			\item $\forall u,v\in T,\quad (u,v)\notin E$
			\item For $u\in T$, let $n_u=\big|\left\{v\in V:\phi(v)=u\right\}\big|$, then we have:
			\begin{equation}
				\sum_{u\in T}\binom{n_u}{2}\leq \big|E\big|
			\end{equation}
		\end{itemize}
	\end{lemma}
	\begin{proof}
		Without loss of generality, we assume the vertices $v_1,v_2,\ldots,v_n$ are sorted in increasing order by degree, i.e., $d(v_1)\leq d(v_2)\leq \ldots\leq d(v_n)$. We can construct $T$ and $\phi$ in the following way:\par
		\begin{algorithm}
			\caption{Construction of $T$ and $\phi$}
			\textbf{Initialization}: $T=\varnothing$\\
			\textbf{Processing}: For $i=1,2,\ldots,n$
			\begin{itemize}
				\item If $\exists u\in T$, s.t. $(u,v_i)\in E$, then we let 
				\begin{equation}\phi(v_i)=\mathop{argmin}_{u\in T,(u,v_i)\in E} d(u)\end{equation}
				\item If such $u$ does not exist, we let $T=T\cup\left\{v_i\right\}$ and $\phi(v_i)=v_i$
			\end{itemize}
		\end{algorithm}
		According to the construction above, we can guarantee the first two requirements in the lemma. Furthermore, since the degrees are sorted in increasing order, we have $d(\phi(u))\geq d(u),\forall u\in V$. Thus
		\begin{equation}
		\begin{split}
			\left|E\right|=&\frac{1}{2}\sum_{v\in V}d(v)=\frac{1}{2}\sum_{u\in T}\left(\sum_{\phi(v)=u} d(v)\right)\\
			\geq&\frac{1}{2}\sum_{u\in T} d(u)\left(d(u)+1\right)\\
			\geq&\sum_{u\in T}\binom{n_u}{2}
			\end{split}
		\end{equation}
	\end{proof}
	(The algorithm described in the lemma is designed in assistance to the theoretical analysis, and does not need to be actually run.)\par
	With the lemma in hand, we are ready to handle the $N_{\beta}$ near pairs in the space, based on which we can lower bound the distance to farther points.
	\begin{theorem}\label{packing}
		For points $x_1,x_2,\ldots,x_n\in \mathbb{R}^d$, $N_{\beta}=\big|\left\{(x_i,x_j):\Vert x_i-x_j\Vert_2\leq \beta r \right\}\big|$. Then $\forall x_0\in \mathbb{R}^d$, we have:
		\begin{equation}\mathop{max}_{1\leq i\leq n}\Vert x_i-x_0\Vert_2\geq \frac{1}{2}\left(\left(\frac{n^2}{2N_{\beta}+n}\right)^{\frac{1}{d}}-1\right)\beta r\end{equation}
	\end{theorem}
	\begin{proof}
		We construct $G=\langle V,E\rangle$ with $V=\left\{v_1,v_2,\ldots,v_n\right\}$ for any pair of vertices $v_i,v_j\in V$, we set $(v_i,v_j)\in E$ if and if only $\Vert x_i-x_j\Vert_2\leq \beta r$. According to Lemma \ref{graphshrink}, we have set $T=\left\{v_{i_1},v_{i_2},\ldots v_{i_t} \right\}$ and mapping $\phi$ for this graph. Let the points in $\mathbb{R}^d$ corresponding to $T$ be $Q=\left\{x_{i_1},x_{i_2},\ldots x_{i_t} \right\}$. By definition we know every two points in $Q$ have distance at least $\beta r$.\par
		According to Cauchy-Schwartz inequality, we have:
		\begin{equation}
		\begin{split}
		n^2=&\left(\sum_{u\in T}n_u\right)^2\leq \left(\sum_{u\in T}n_u^2\right)\left(\sum_{u\in T}1\right)\\
		=&\big|T\big|\cdot \left(n+2\sum_{u\in T}\binom{n_u}{2}\right)
		\leq \big|T\big|(n+2\big|E\big|)
		\end{split}
		\end{equation}
		So we have $\big|Q\big|=\big|T\big|\geq \frac{n^2}{n+2N_{\beta}}$.\par
		Let $r^*=\mathop{max}_{1\leq i\leq n}\Vert x_i-x_0\Vert_2$, consider a ball $B_0=B(x_0,r^*+\frac{\beta r}{2})$, centered at $x_0$ with radius $r^*+\frac{\beta r}{2}$. And for each $\forall x_i, 1\leq i\leq n$, let $d$-dimensional ball $B_i=B(x_i,\frac{\beta r}{2})$. Since $\Vert x_i-x_0\Vert_2\leq r^*$, we have $\bigcup_{i=1}^{n}B_i\subset B_0$. On the other hand, since the points in $Q$ are at least $\beta r$ distant from each other, we have $B_i\cap B_j=\varnothing,\forall x_i,x_j\in Q, i\neq j$. Thus we have:
		\begin{equation}
			\begin{split}
				&\frac{\pi^{\frac{d}{2}}}{\Gamma\left(1+\frac{d}{2}\right)}\left(r^*+\frac{\beta r}{2}\right)^d=Vol\left(B\left(x_0,r^*+\frac{\beta r}{2}\right)\right)\\
				\geq& Vol\left(\bigcup_{x'\in Q}{B(x', \frac{\beta r}{2})}\right)
				=\sum_{x'\in Q}{Vol\left(B(x',\frac{\beta r}{2})\right)}
				=\big|Q\big|\cdot \frac{\pi^{\frac{d}{2}}}{\Gamma\left(1+\frac{d}{2}\right)} \left(\frac{\beta r}{2}\right)^d
			\end{split}
		\end{equation}
		By plugging in the lower bound for $\big|Q\big|$ we get:
		\begin{equation}r^*\geq \left(\left(\frac{n^2}{n+2N_{\beta}}\right)^{\frac{1}{d}}-1\right)\frac{\beta r}{2}\end{equation}
	\end{proof}
	This bound is informative only for relatively low dimensionality. Otherwise, for example, if $d=\omega(log n)$, we will have $\lim_{n\rightarrow +\infty}\left(\frac{n^2}{n+2N_{\beta}}\right)^{\frac{1}{d}}=1$, and the bound converges to zero.  Actually, the geometric structure of $\mathbb{R}^d$ can tell us little information when $d=\omega(\log n)$ and no other constraints are posed. Indeed, the classical analysis for LSH is tight for this case. On the other hand, there are still much we can do in high dimensions: by standard Johnson-Lindenstrauss argument we can restrict dimensionality at the order of $O(\log n)$ while preserving locality; we can also replace the dimensionality of space with doubling dimension, as discussed in Section 5.\par
	\section{Dimension-dependent Refined Analysis for LSH}
	In section 3 we have already proposed a tighter estimation for the distance from query point to the farthest point. This result can be applied in the analysis of locality sensitive hashing, and yield an essentially sharper bound for it.\par
	Given points $x_1,x_2,\ldots,x_n\in\mathbb{R}^d$, we can build hash tables using Algorithm \ref{framework}, whose parameters $K, L$ will be set to optimize the bound later on. When a query point arrives, as the analysis in Section 2, we have:
	\begin{equation}\mathop{E}_{h\sim\mathcal{H}}\left[\vert S_i\vert\right]\leq1+\sum_{\Vert x_i-x_0\Vert_2\geq \alpha r}\exp\left\{\frac{K\cdot\log p(1)}{\rho\left(\frac{\Vert x_i-x_0\Vert_2}{r}\right)}\right\}
	\end{equation}
	Based on the sphere packing results, we are ready to upper bound the value of big summation above.
	\begin{lemma}\label{summation}
		Given $x_1,x_2,\ldots,x_n\in \mathbb{R}^d$, fixed parameters $\alpha\geq 1,\beta>0, \eta>0, p\in(0,1)$, $\rho(\cdot)$ be a monotonic increasing function on $(1,+\infty)$, and let $N_{\beta}=\big|\left\{(i,j):\Vert x_i-x_j\Vert_2\leq \beta r\right\}\big|$, then we have:
		\begin{equation}
		\begin{split}
		&\sum_{\Vert x_i-x_0\Vert_2\geq \alpha r}\exp\left\{\frac{\log p}{\rho\left(\frac{\Vert x_i-x_0\Vert_2}{r}\right)}\right\}\\
		\leq& p^{\frac{1}{\rho(\alpha)}}\left(1+\frac{2(\alpha+\eta)}{\beta}\right)^{\frac{d}{2}}\sqrt{2N_{\beta}+n}+p^{\frac{1}{\rho(\alpha+\eta)}}n
		\end{split}
		\end{equation}
	\end{lemma}
	In estimating the summation, classical analyses roughly divide the points according to their distances from query point, at threshold $\alpha r$. In Lemma~\ref{summation}, we divide further at threshold $(\alpha+\eta)r$ to get more accurate bounds. The detailed proof for this lemma is deferred to Appendix. We will see from the later analysis that this bound could not be asymptotically improved by more precise dividing, since the points outside $(\alpha+\eta)r$ distance only make tiny contribution to the sum.
	
	By applying the standard strategies for the parameter setting in LSH, here follows our dimension-dependent query time bound, the detailed proof is deferred to Appendix.
	\begin{theorem}\label{main}
		(Main theorem for real dimension)\par
		For $x_1,x_2,\ldots,x_n\in\mathbb{R}^d$, if we use a uniform locality sensitive hashing $\mathcal{H}$ with $\rho=\rho(\alpha)+o(1)$ and $\tau$ computation cost to solve $\alpha$-approximate $r$-near neighbor problem with probability $1-\delta$, then $\forall \beta>0$, with proper parameter selection, the expected query time is upper bounded by: \begin{equation}O\left(\left(1+\frac{2\mu}{\beta}\right)^{\frac{d}{2}\rho(\alpha)}(N_{\beta}+n)^{\frac{1}{2}\rho(\alpha)+o(1)}\tau\log\frac{1}{\delta}\right)\end{equation}
		where $\mu$ satisfies $\rho(\mu)=\frac{1}{2}\rho(\alpha)$ for monotonic decreasing function $\rho(\cdot)$
	\end{theorem}
	
	We apply this theorem to the two famous locality sensitive hashing schemes. Since their constructions do not involve parameter $\alpha$ to be known in advance, they actually satisfy the uniform locality sensitive property proposed in Section 2.
	\begin{proposition}\label{weaklsh}
		For the line projection LSH in \cite{DBLP:conf/stoc/IndykM98}, we have $\rho(\alpha)=\frac{1}{\alpha}$, the expected running time is bounded with:
		\begin{equation}O\left(d\left(1+\frac{4\alpha}{\beta}\right)^{\frac{d}{2\alpha}}(N_{\beta}+n)^{\frac{1}{2\alpha}}\log\frac{1}{\delta}\right),\forall \beta>0\end{equation}
	\end{proposition}
	\begin{proposition}\label{stronglsh}
		For the random grid of ball LSH in \cite{DBLP:conf/focs/AndoniI06}, we have $\rho(\alpha)=\frac{1}{\alpha^2}+o(1)$, the expected running time is bounded with:
		\begin{equation}O\left(d\left(1+\frac{2\sqrt{2}\alpha}{\beta}\right)^{\frac{d}{2\alpha^2}}(N_{\beta}+n)^{\frac{1}{2\alpha^2}+o(1)}\log\frac{1}{\delta}\right),\forall \beta>0\end{equation}
	\end{proposition}
	This bound on expectation holds uniformly for $\forall \beta>0$, and the optimal $\beta$ can be selected to minimize the bound. Actually, if the points are well-dispersed, i.e., $N_{\beta}$ increases slowly with $\beta$, this bound leads to significantly lower running time for LSH. The classical analysis can be seen as a special case of this dimensionality-dependent analysis, for we have:
	\begin{equation}\lim\limits_{\beta\rightarrow+\infty}\left(1+\frac{2\mu}{\beta}\right)^{\frac{d}{2}\rho(\alpha)}(N_{\beta}+n)^{\frac{1}{2}\rho(\alpha)+o(1)}=n^{\rho(\alpha)+o(1)}\end{equation}
	Despite its sharpness, the major drawback of this bound is that it depends exponentially on the data dimensionality, if the parameter $\alpha$ and $\beta$ are constants. We will explain in the next section on how to overcome it, by further exploiting intrinsic structure of the points.
	\section{Doubling Metric Counterparts}
	Since the analyses in previous sections are based upon the point sets' rate of expansion, instead of their real dimensionality, it would naturally generalize to the case of dimensionality intrinsic in the data and provide a better bound. In this section, we extend our analysis to doubling dimension case. On the one hand, doubling dimension can appropriately capture the phenomena that high-dimensional data are usually lying approximately on a low-dimensional manifold; on the other hand, bounds based on doubling dimension are also applicable to more general metric spaces.\par
	The notion of doubling dimension has been studied in a wide range of literature. There are several different definitions for doubling dimension\cite{DBLP:conf/stoc/Talwar04,DBLP:conf/soda/KrauthgamerL04,DBLP:conf/stoc/KargerR02}. They are equivalent except for an absolute constant factor. Here we adopt the definition in \cite{DBLP:conf/soda/KrauthgamerL04}, for it not only adapts more naturally to our problem, but also generalizes to metrics other than Euclidean distances.
	\begin{description}
		\item[Doubling dimension] The doubling dimension of a metric $X$ is the minimal $d_0$ such that $\forall Y\subseteq X$, there exists a series of subsets $\left\{Y_i\right\}_{i=1}^{2^{d_0}}$, such that $Y\subseteq \bigcup_{i=1}^{2^{d_0}}Y_i$ and 
		\begin{equation}
		\begin{split}
		\mathop{max}_{y,y'\in Y_i}\Vert y-y'\Vert_X\leq& \frac{1}{2}\mathop{max}_{y,y'\in Y}\Vert y-y'\Vert_X\\
		\forall i=&1,2,\ldots 2^{d_0}
		\end{split}
		\end{equation}
	\end{description}
	The following lemma from \cite{DBLP:conf/soda/KrauthgamerL04} characterizes the packing properties for metrics with doubling dimensions, and are widely applied in various problems:
	\begin{lemma}\label{doubling}
		For a metric $X$ with doubling dimension $d_0$, then for any finite subset $Y\subseteq X$ we have
		\begin{equation}\lceil\log \frac{\mathop{max}_{y,y'\in Y}{\Vert y-y'\Vert_2}}{\mathop{min}_{y,y'\in Y}{\Vert y-y'\Vert_2}}\rceil\geq\frac{1}{d_0}\log \big|Y\big|\end{equation}
	\end{lemma}
	
	This helps us to establish the doubling-dimension version of Theorem \ref{packing}.
	\begin{lemma}\label{doublingpacking}
		For points $x_1,x_2,\ldots,x_n\in \mathbb{R}^d$ with doubling dimension $d_0$, for $\beta>0$, let $N_{\beta}=\big|\left\{(x_i,x_j):\Vert x_i-x_j\Vert_2\leq \beta r \right\}\big|$. Then $\forall x_0\in \mathbb{R}^d$, we have:
		\begin{equation}\mathop{max}_{1\leq i\leq n}\Vert x_i-x_0\Vert_2\geq \frac{1}{4}\left[\left(\frac{n^2}{2N_{\beta}+2n}\right)^{\frac{1}{d_0+1}}-1\right]\beta r\end{equation}
	\end{lemma}
	
	By plugging this result into the proof of Theorem \ref{main}, we get:
	\begin{theorem}\label{doublingmain}
		(Query time bound for doubling dimension)\par
		For $x_1,x_2,\ldots,x_n\in\mathbb{R}^d$ with doubling dimension $d_0$, if we use a uniform locality sensitive hashing $\mathcal{H}$ with $\rho=\rho(\alpha)+o(1)$ and $\tau$ computation cost to solve $\alpha$-approximate $r$-near neighbor problem with probability $1-\delta$, then $\forall \beta>0$, with proper parameter selection, the expected query time is upper bounded by: \begin{equation}O\left(\left(1+\frac{4\mu}{\beta}\right)^{\frac{d_0+1}{2}\rho(\alpha)}(N_{\beta}+n)^{\frac{1}{2}\rho(\alpha)+o(1)}\tau\log\frac{1}{\delta}\right)\end{equation}
		where $\mu$ satisfies $\rho(\mu)=\frac{1}{2}\rho(\alpha)$ for monotonic decreasing $\rho(\cdot)$
	\end{theorem}
	Similar results as Proposition \ref{weaklsh} and Proposition \ref{stronglsh} also holds for doubling dimension. Furthermore, Theorem~\ref{doublingmain} also implies improved bounds for well-dispersed points in general metric spaces. For example, by applying Theorem \ref{doublingmain} to  LSH based on $p$-stable distribution \cite{DBLP:conf/compgeom/DatarIIM04}, we can get the following proposition for general $\ell_p$(the $p$-stable distribution guarantees for $p\in(0,2]$, and we need $p\geq 1$ to guarantee the triangle inequality):
	\begin{proposition}\label{lplsh}
		$\forall p\in[1,2]$, consider the $p$-stable distribution LSH in \cite{DBLP:conf/compgeom/DatarIIM04} which solves approximate near neighbor in $\ell_p$ metric. Assuming the doubling dimension of $\langle X,\Vert\cdot\Vert_p\rangle$ to be $d_0$, we have $\rho(\alpha)=\frac{1}{\alpha}+o(1)$, $\forall \beta>0$, with proper parameter selection, the expected query time is bounded with:
		\begin{equation}O\left(d\left(1+\frac{8\alpha}{\beta}\right)^{\frac{d_0+1}{2\alpha}}(N_{\beta}+n)^{\frac{1}{2\alpha}+o(1)}\log\frac{1}{\delta}\right)\end{equation}
	\end{proposition}
	\section{When do our bounds work?}
	In this section, we will discuss the parameters in our bounds and show that, in many natural scenarios, where data points are well-dispersed or doubling dimension is low, are bound can significantly improve over classical worst-case bounds for LSH.\par
	In our analyses, the parameter $\beta$ should be chosen carefully in order to minimize our bounds. We are particularly interested in the case where $N_{\beta}$ is approximately at the same order as $n$. Specifically, for a small constant $\epsilon>0$, let $C_{\epsilon}(n)=\sup\{\beta: N_{\beta}<n^{1+\epsilon}\}$. Well-dispersed data will result in larger value of $C_{\epsilon}$.\par
	In the low-dimensional case, namely, $d=o(\log n)$, we use the bound in Theorem~\ref{main}; otherwise we will turn to doubling dimension bound in Theorem~\ref{doublingmain}. For a finite metric space, we have $d_0\leq \log n$ by definition. Let $d_0=\xi \log n$ be the doubling dimension of data set, with $\xi\in(0,1]$.
	Combined with optimal data-independent LSH, we get the following expected query time bound:
	\begin{equation}\frac{\log E[T_{query}]}{\log n}\leq \frac{1}{2\alpha^2}\left(1+\epsilon+\xi\log\left(1+\frac{4\sqrt{2}\alpha}{C_{\epsilon}(n)}\right)\right)+o(1)\end{equation}
	Our bounds are informative in the following scenarios:
	\begin{itemize}
		\item Data points are significantly well-dispersed, i.e., $\lim\limits_{n\rightarrow\infty}C_{\epsilon}(n)=\infty$ for some small $\epsilon$. For this case we have $E[T_{query}]\leq n^{\frac{1+\epsilon}{2\alpha^2}+o(1)}$, since $\xi$ is bounded by one. An example for this condition is a set of sparse vectors in high dimensions, i.e. $d=poly(n)$, where the dimension gets higher when $n$ gets larger, and the data points become more dispersed.
		\item Spaces with significantly low (doubling) dimensions, i.e., $\lim\limits_{n\rightarrow\infty}\xi=0$. For this case our bound becomes $E[T_{query}]\leq n^{\frac{1}{2\alpha^2}+o(1)}$. Though the bound is still polynomial in $n$, it implies that LSH can automatically adapt to low-dimensional space or doubling metrics.
		\item For general case, we assume $C$ and $\xi$ are both $\Theta(1)$. The quality of our bound depends on whether $1+\frac{4\sqrt{2}\alpha}{C}<2^{\frac{1}{\xi}}$ holds. Larger $C$ and smaller $\xi$ will be helpful.
	\end{itemize}
	For the above three cases, our analysis is uniformly sharper than all existing data-oblivious LSH results. It is also sharper than the best data-dependent LSH in the first two cases. Interestingly, our results guarantee sub-linear query time even when $\alpha=1$, while all existing worst-case bounds for LSH query time are informative only for $\alpha>1$. This result implies that LSH can take advantage of well-dispersed data, even if we want to perform exact $r$-near neighbor search.\par
	Furthermore, our analysis shows that the parameter choice designed for worst-case performance can be suboptimal in practice, when the data points satisfy dispersion or doubling dimension structure. Specifically, since the collision probabilities are usually overestimated, bucket sizes can be smaller than needed, namely, $K$ is often set too large in the worst-case-optimal setting. To make use of our bounds in practice, we may incorporate empirical estimates for $N_{\beta}$ and prior knowledge about doubling dimension, and plug into the choice of $K$ in Theorem~\ref{main}.
	\section{Conclusion and Open Questions}
	In this paper, we present a refined analysis for query time of approximate near neighbor via LSH, given well-dispersed data points. Though the previous analyses are tight for worst case, they could not explain how LSHs work better in real data sets where some structures are assumed. We address this issue by introducing $N_{\beta}$, the number of $\beta r$-near pairs, to describe the dispersion of data points. Using a generalized version of sphere packing argument, we present an $O\left(\left(1+\frac{2\mu}{\beta}\right)^{\frac{d}{2}\rho(\alpha)}(N_{\beta}+n)^{\frac{1}{2}\rho(\alpha)+o(1)}\right)$ upper bound for expected query time of LSH, based on proper choice of parameters. We also generalize this result to the case of doubling dimension, and obtained an $O\left(\left(1+\frac{4\mu}{\beta}\right)^{\frac{d_0+1}{2}\rho(\alpha)}(N_{\beta}+n)^{\frac{1}{2}\rho(\alpha)+o(1)}\right)$ bound. Compared with other existing results, these bounds make essential improvements when the data points are dispersed well or have low doubling dimensions.\par
	There are still many problems on explaining the performance of LSHs that are left open:
	\begin{itemize}
		\item In our analyses, the relaxation of inequalities are actually loose, in terms of the constant factor before $\frac{\alpha}{\beta}$. We intend to give a tighter bound in the future, so that the dependence on dispersion parameters can be further relaxed.
		\item The data-dependent schemes\cite{DBLP:conf/stoc/AndoniR15} could not be applied to our analyses directly, since they are not uniform. A more precise analyses on the geometric structures of high-dimensional data sets will shed lights on data-dependent LSHs.
		\item We will also seek quantities other than $N_{\beta}$ that describes characteristics and structures of data points, which can explain the performance of near neighbor algorithms.
	\end{itemize}
	\bibliographystyle{plain}
	\bibliography{ref}
	\section*{Appendix: Deferred Proofs}
	\begin{proof}(Proof of Lemma~\ref{summation})
		Let $x_{p(1)},x_{p(2)},\ldots,x_{p(n)}$ be a permutation of $x_1,\ldots,x_n$ s.t. $\left\{\Vert x_{p(i)}-x_0\Vert_2\right\}_{i=1}^{n}$ is sorted in ascending order. For $\forall k\in\left\{1,2,\ldots,n\right\}$ and $\beta>0$, apparently we have \begin{equation}\big|\left\{(i,j):\Vert x_{p(i)}-x_{p(j)}\Vert_2\leq \beta r,1\leq i<j\leq k\right\}\big|\leq N_{\beta}\end{equation}
		Plugging it into Theorem \ref{packing}, we have:
		\begin{equation}
		\begin{split}
		\Vert x_{p(k)}-x_0\Vert_2=\mathop{max}_{1\leq i\leq k}\Vert x_{p(i)}-x_0\Vert_2
		\geq \left(\left(\frac{k^2}{k+2N_{\beta}}\right)^{\frac{1}{d}}-1\right)\frac{\beta r}{2}
		\geq \left(\left(\frac{k^2}{n+2N_{\beta}}\right)^{\frac{1}{d}}-1\right)\frac{\beta r}{2}
		\end{split}
		\end{equation}
		Let $k_0=\mathop{max}\left\{k: \Vert x_{p(k)}-x_0\Vert_2\leq (\alpha+\eta) r\right\}$, according to the inequality above, we have:
		\begin{equation}
		\begin{split}
		\left(\left(\frac{k_0^2}{n+2N_{\beta}}\right)^{\frac{1}{d}}-1\right)\frac{\beta r}{2}\leq (\alpha+\eta) r\Rightarrow k_0\leq \left(1+\frac{2(\alpha+\eta)}{\beta}\right)^{\frac{d}{2}}\sqrt{n+2N_{\beta}}
		\end{split}
		\end{equation}
		Then we can easily lower bound $\Vert x_{p(i)}-x_0\Vert_2$ with $\alpha r$ for $i\leq k_0$ and bound with $(\alpha+\eta)r$ for $i>k_0$, and get the following result:
		\begin{equation}
		\begin{split}
		\sum_{\Vert x_i-x_0\Vert_2\geq \alpha r}\exp\left\{\frac{\log p}{\rho\left(\frac{\Vert x_i-x_0\Vert_2}{r}\right)}\right\}
		\leq p^{\frac{1}{\rho(\alpha)}}\left(1+\frac{2(\alpha+\eta)}{\beta}\right)^{\frac{d}{2}}\sqrt{2N_{\beta}+n}+p^{\frac{1}{\rho(\alpha+\eta)}}
		\end{split}
		\end{equation}
	\end{proof}
	\begin{proof}(Proof of Theorem~\ref{main})
		We use the algorithmic framework described in Algorithm \ref{framework}, the parameters $K, L$ will be later set to optimize the bound.\par
		Let $p(s)=Pr_{h\sim \mathcal{H}}\left\{h(x)=h(y)\big| \Vert x-y\Vert_2=sr\right\}$ be the collision probability induced by $\mathcal{H}$. As discussed in Section 1, it suffices to set $L=\frac{1}{p(1)^K}\log \frac{1}{\delta}$ to guarantee the collision probability. Then the expected query time is bounded with:
		\begin{equation}
		\frac{T(n)}{\tau\log \frac{1}{\delta}}\leq\frac{1}{p(1)^K}\left(1+\sum_{\Vert x_i-x_0\Vert_2\geq \alpha r}\exp\left\{\frac{K\cdot\log p(1)}{\rho\left(\frac{\Vert x_i-x_0\Vert_2}{r}\right)}\right\}\right)
		\end{equation}
		By applying Lemma \ref{summation} with $p=p(1)^K$, and let $M=\left(1+\frac{2(\alpha+\eta)}{\beta}\right)^{\frac{d}{2}}\sqrt{2N_{\beta}+n}$, we get:
		\begin{equation}
		\begin{split}
		\frac{T(n)}{\tau \log{\frac{1}{\delta}}}\leq  p(1)^{-K}&+ p(1)^{K(\frac{1}{\rho(\alpha)+o(1)}-1)}M\\
			&+p(1)^{K(\frac{1}{\rho(\alpha+\eta)+o(1)}-1)}n
		\end{split}
		\end{equation}
		We set the parameters as follows:\par 
		\begin{equation}K=-\frac{\rho(\alpha)\log M}{\log p(1)},\quad \eta=\mu-\alpha\end{equation}
		From our choice of parameters, we have $p(1)^K=\frac{1}{M^{\rho(\alpha)}}$, and
		\begin{equation}
		M\cdot p(1)^{K\left(\frac{1}{\rho(\alpha)+o(1)}-1\right)}=M^{\rho(\alpha)+o(1)}
		\end{equation}
		\begin{equation}
			\begin{split} 		
			n\cdot p(1)^{K\left(\frac{1}{\rho(\alpha+\eta)+o(1)}-1\right)}
			=nM^{\rho(\alpha)-\frac{\rho(\alpha)}{\rho(\mu)}+o(1)}\leq M^{\rho(\alpha)+o(1)}
			\end{split}
		\end{equation}
		Thus the query time is upper bounded by $O(M^{\rho(\alpha)+o(1)}\tau\log\frac{1}{\delta})$.
	\end{proof}
	\begin{proof}(Proof of Lemma~\ref{doubling})\\
		By definition we can construct a series of subsets $\left\{Y_i\right\}_{i=1}^{2^{d_0}}$ satisfying the doubling metric condition, such that $Y\subseteq \bigcup_{i=1}^{2^{d_0}}Y_i$. Apparently we have:
		\begin{equation}\big|Y\big|\leq \sum_{i=1}^{2^{d_0}}\big|Y_i\big|\end{equation}
		Each $Y_i$ will be a subset with at most half of the radius. We can then recursively divide each $Y_i$ until there are at most 2 points contained in each set, and by simple calculation we get this result.
	\end{proof}
	\begin{proof}(Proof of Lemma~\ref{doublingpacking})\\
		For $\beta>\frac{max\Vert x_i-x_j\Vert_2}{r}$, the inequality is trivial since RHS becomes negative.\par
		For $\beta\in\left(0,\frac{max\Vert x_i-x_j\Vert_2}{r}\right)$, let $X=\left\{x_0,x_1,\ldots,x_n\right\}$. By adding just one point, the doubling dimension will not exceed $d_0+1$. Since adding $x_0$ to the point set will increase the number of near pairs up to $n$, we have $\big|\left\{(i,j):0\leq i<j\leq n:\Vert x_i-x_j\Vert_2\leq \beta r\right\}\big|\leq N_{\beta}+n$. According to Lemma \ref{graphshrink}, we obtain a set of points $T$ with $\big|T\big|\geq \frac{n^2}{2n+2N_{\beta}}$ and $\forall x,x'\in T, \Vert x-x'\Vert_2\geq \beta r$. The doubling dimension of $T$ does not exceed $d_0+1$. By applying Lemma \ref{doubling} to $T$, we get:
		\begin{equation}
		\begin{split}
		\frac{1}{d_0+1}\log \big|T\big|
		\leq\lceil\log \frac{\mathop{max}_{y,y'\in T}{\Vert y-y'\Vert_2}}{\beta r}\rceil
		\leq 1+\log \frac{\mathop{max}_{y,y'\in T}{\Vert y-y'\Vert_2}}{\beta r}
		\end{split}
		\end{equation}
		Thus we have:
		\begin{equation}\mathop{max}_{y,y'\in T}{\Vert y-y'\Vert_2}\geq \frac{\beta r}{2}\big|T\big|^{\frac{1}{d_0+1}}\geq \left(\frac{n^2}{2n+2N_{\beta}}\right)^{\frac{1}{d_0+1}}\frac{\beta r}{2}\end{equation}
		On the other hand, according to the triangle inequality, we have
		\begin{equation}2\mathop{max}_{1\leq i\leq n}\Vert x_i-x_0\Vert_2\geq \mathop{max}_{y,y'\in T}{\Vert y-y'\Vert_2}\end{equation}
		By putting them together, we get:
		\begin{equation}
		\begin{split}
		\mathop{max}_{1\leq i\leq n}\Vert x_i-x_0\Vert_2
		\geq \frac{1}{4}\left[\left(\frac{n^2}{2N_{\beta}+2n}\right)^{\frac{1}{d_0+1}}\right]\beta r 
		\geq \frac{1}{4}\left[\left(\frac{n^2}{2N_{\beta}+2n}\right)^{\frac{1}{d_0+1}}-1\right]\beta r
		\end{split}
		\end{equation}
	\end{proof}
\end{document}